\documentclass{llncs}
\usepackage{amsfonts}
\usepackage{amsmath} 
\usepackage{amssymb}
\usepackage{enumerate}

\newcommand{\QED}{\end{proof}}

\newcommand{\N}{\mathbb{N}}

\newcommand{\R}{\mathbb{R}}

\newcommand{\K}{{\mathrm{K}}}

\title{Predictive Complexity and Generalized Entropy Rate of
  Stationary Ergodic Processes}
\author{Mrinalkanti Ghosh \and Satyadev Nandakumar}
\institute{
Department of Computer Science and Engineering,\\
Indian Institute of Technology Kanpur,\\
Kanpur, U.P., India.
}

\begin{document}
\maketitle

\begin{abstract}
In the online prediction framework, we use generalized entropy of to
study the loss rate of predictors when outcomes are drawn according to
stationary ergodic distributions over the binary alphabet. We show
that the notion of generalized entropy of a regular game \cite{KVV04}
is well-defined for stationary ergodic distributions. In proving this,
we obtain new game-theoretic proofs of some classical information
theoretic inequalities. Using Birkhoff's ergodic theorem and
convergence properties of conditional distributions, we prove that a
classical Shannon-McMillan-Breiman theorem holds for a restricted
class of regular games, when no computational constraints are imposed
on the prediction strategies.

If a game is mixable, then there is an optimal aggregating strategy
which loses at most an additive constant when compared to any other
lower semicomputable strategy. The loss incurred by this algorithm on
an infinite sequence of outcomes is called its \emph{predictive
  complexity}. We use our version of Shannon-McMillan-Breiman theorem
to prove that when a restriced regular game has a predictive
complexity, the predictive complexity converges to the generalized
entropy of the game almost everywhere with respect to the stationary
ergodic distribution.
\end{abstract} 

\section{Introduction}
We consider the online prediction question studied by
\cite{VW98},\cite {VV02}, \cite{KVV04}, \cite{Fortnow:PD},
\cite{KVV07} in the setting of a stationary stochastic process. In
this setting, we have a sequence of outcomes $x_0, x_1, \dots$ from a
finite alphabet. A predictor, given the history up to a certain index,
predicts what the next outcome will be. We allow the predictor to
present its prediction as a convex combination which represents the
weight it assigns to each outcome in the alphabet. The game proceeds
by revealing the next outcome, and then asking for the prediction of
the future outcome. For an overview of this area, see
\cite{CBL06}. Independently, Merhav and Feder \cite{MeFe98}, Feder
\cite{Fede91} and Feder et. al. \cite{FeMeGu92} have studied the
question of optimal finite-state predictors with respect to Shannon
entropy, in the setting of stationary Markov Chains. It is known that
the log-loss game characterizes Shannon entropy. The present line of
work generalizes their approach in two ways - first, in considering
loss functions besides log-loss, and second, in considering optimal
processes over stationary ergodic distributions.

A natural question in this context is how well the predictor is doing
as the game progresses. We measure the discrepancy between the actual
outcome and the predicted one, with a \emph{loss function}. This helps
us to ask whether \emph{optimal} predictors exist - those which incur
at most the same loss as as any other predictor on any outcome,
ignoring additive constants. Indeed if such an optimal predictor
exists, we can use its loss rate on a particular sequence of outcomes
to define its inherent \emph{predictability} (see for example,
\cite{VW98}, \cite{VV02}). 

Besides competitive advantage above other predictors, we can also
characterize the performance of an optimal predictor by examining its
expected loss assuming the outcomes are drawn from a particular
distribution. Prior work by Kalnishkan et al. \cite{KVV04} establishes
that if the outcomes are drawn independently according to a Bernoulli
distribution on the alphabet, then the expected loss rate of an
optimal predictor is the \emph{generalized entropy} \cite{GD04} of the
loss function. In this paper, we extend this result to the important
setting of stationary ergodic distributions.

The contributions of our paper are threefold.

\begin{enumerate}
\item First, we show that the generalized entropy rate of a stationary
  ergodic process is well-defined, if the game is \emph{regular}. We
  provide ``game-theoretic'' proofs of classical information-theoretic
  inequalities, giving new intuitive proofs even in the special case
  of the Shannon entropy. This constitues sections 3 and 4 of the
  paper.

\item Second, under a continuity and an integrability constraint, we
  show that optimal strategies exist for regular games.\footnote{There
    is an independent characterization of games with optimal
    strategies in terms of convexity of loss-regions \cite{KVV07}. We
    deal with this approach in the final section of our paper.}  We
  show that the loss rate incurred by such a strategy is the
  generalized entropy rate of the stationary ergodic process. This is
  a Shannon-McMillan-Breiman theorem for generalized entropy. This
  result is new, and we provide a proof using Vitali Convergence. This
  constitutes section 5 of the paper.


\item Using the above results, we show that when a game has
  \emph{predictive complexity}, an optimal aggregator algorithm
  attains the entropy rate of the game.

  The proof that the aggregator incurs at most the entropy rate of
  loss crucially uses our Shannon-McMillan-Breiman Theorem.

 The proof that the aggregator incurs at least the entropy rate of
 loss uses some properties of stationary ergodic processes that we
 prove in Sections 3 and 4. This constitutes the final section of the
 paper.
\end{enumerate}

\section{Preliminaries}
As defined in \cite{KVV04}, a game $\mathcal{G}$ is a triple
$(\Sigma,\Gamma,\lambda)$ where $\Sigma$ is a finite alphabet space,
$\Gamma$ is the space of predictions and
$\lambda:\Sigma\times\Gamma\to[0,\infty]$ is the loss function, to be
defined below. We will only consider the binary alphabet in this
paper. 

Intuitively, we model a predictor function which, given the string of
outcomes so far, will predict the next outcome. We consider a slightly
general framework where the predictor does not have to necessarily
predict only one outcome. It is allowed to output a point $(p_0,p_1)
\in \Gamma^2$ (equivalently, a probability vector, where $p_0$ is the
predicted probability that the next bit is 0, and $p_1$, the
probability that the next bit is 1). The game proceeds by revealing
the next outcome. Let this outcome be $b$. The prediction strategy is
said to incur the loss $\lambda(b, (p_0,p_1))$.

As is customary, we adopt the notation $\N$ for the set of natural
numbers, starting from $0$. The set of strings of length $n$ is
denoted $\Sigma^n$. The set of finite binary strings is denoted
$\Sigma^*$ and the set of infinite binary sequences is denoted
$\Sigma^\infty$. For a finite or an infinite sequence $x$, the
notation $x^j_i$ denotes $x_i \dots x_j$. If $x$ is shorter than $n$
bits, $x^{n-1}_0$ denotes $x$ itself. If $x$ is a finite string, and
$\omega$ is a finite string or an infinite sequence, then $x \cdot
\omega$ denotes the result of concatenating $\omega$ to $x$. For each
natural number $i$, let $\Pi^i$ be the class of all functions mapping
$i$-long strings to $\Gamma$.

We call a family of functions $\wp$ a \emph{strategy} if $\forall i\in
\mathbb{N},|\wp \cap \Pi^i|=1$, i.e, there is unique function which
takes an $i$-length string as input and produce a strategy based on
the input. We call that function $\wp^i$. Thus the prediction strategy
is a non-uniform family. We impose no computational constraints until
the final part of the paper. 

\section{Loss functions}

The generalized entropy of a game is defined in terms of convex loss
functions described above. We define the losses incurred by a strategy
on a finite string $w$ of outcomes, as the cumulative loss that it
incurs on each bit of $w$. This follows the definition given in
\cite{KVV04} and \cite{KVV07}. We generalize the notion slightly to
deal with the expected loss that a strategy incurs with respect to a
stationary distribution.

\begin{definition}
\label{loss}The \emph{loss} that a prediction strategy $\wp$, incurs
on a finite 
string $w$ of outcomes is defined to be
$$\text{Loss}(w,\wp)=\sum_{i=0}^{|w|-1}\lambda(w_i,\wp^i(\omega^{i-1}_0))$$
\end{definition}

In order to study when a strategy is better than another, we study the
average loss it incurs, when outcomes are drawn from a stationary
distribution. We consider the strategy which incurs the minimal
expected loss on a particular set, if such a strategy exists. Let
$(\Sigma^\infty,\mathcal{F},P)$ be the probability space where
$\mathcal{F}$ is the Borel $\sigma$-algebra generated by cylinders
$$C_x=\{\omega\in\Omega\;\mid\;x\text{ is a prefix of }\omega\}$$ for
all finite strings $x$. and $P:\mathcal{F}\to [0,1]$ is the probability
measure. 

Let $X = (X_0, X_1, \dots)$ be a sequence of random variables on the
probability space - for each $i \in \N$, $X_i$ maps $\Sigma^\infty$ to
$\R$.  For $k \ge 1$, let $S_k X$ denote the sequence $(X_{k},
X_{k+1}, \dots)$ - that is, $X$ ``shifted left'' $k$ times.

\begin{definition} \cite{Shiryaev}
A sequence of random variables $X$ is \emph{stationary} if the
probabilities of $S_k X$ and $X$ coincide for every $k \ge 1$. That
is, for every Borel set $B$ in the $\sigma$-algebra over $\R^\infty$, 
$$P( X \in B) = P(S_k X \in B).$$
\end{definition}

We could also use the terminology of measure-preserving
transformations to capture stationarity. A transformation $T: \Omega
\to \Omega$ is said to be \emph{measure-preserving} if for every $A
\in \mathcal{F}$, $P(T^{-1}A) = P(A)$. A measure-preserving
transformation is said to be \emph{ergodic} if $T^-1(A) = A$ if and
only if $P(A)$ is either 0 or 1. \cite{Bill65}

The class of stationary processes correspond almost exactly to the
class of probability spaces $(\Omega, \mathcal{F}, P, T)$ where $T:
\Omega \to \Omega$ is a $P$-measure-preserving transformation. For $k
\in \N$, let $T^k$ denote the iterated application of $T$ on itself,
$k$ times. It is easy to see that if $T$ is measure preserving and
$X_0$ is a random variable, then $(X_0, X_0 \circ T, X_0 \circ T^2,
\dots)$ is a stationary sequence. We also have the converse.

\begin{lemma}\cite{Shiryaev}
For every stationary sequence $X$ on a probability space $(\Omega,
\mathcal{F}, P)$, there is a probability space $(\tilde\Omega,
\tilde{\mathcal{F}}, \tilde{P})$, a random variable $\tilde X$ and a
$\tilde{P}$-measure preserving transformation
$\tilde{T}:\tilde{\Omega}\to\tilde{\Omega}$ such that the distribution
of $(\tilde{X_0}, \tilde{X_0}\circ\tilde{T},
\tilde{X_0}\circ\tilde{T}^2, \dots)$ coincides with the distribution
of $X$.
\end{lemma}


On an alphabet space, we are interested in the coordinate random
variables $X_i(\omega) = \omega_i\ (i \in \N)$, and any probability
distribution such that $X$ is stationary with respect to it, will be
called a stationary distribution. A probability space with respect to
which the left-shift transformation is ergodic will be called an
\emph{ergodic distribution}.

\begin{definition}
\label{nentropy}
We define the \emph{$n$-step generalized entropy} of the game to be
\begin{align}
\label{eqn:nentropy}
H_n=\displaystyle\inf_{\wp}\sum_{w\in\Sigma^n}P(w)\text{Loss}(w,\wp),
\end{align}
where $(\Sigma^\infty, \mathcal{F}, P)$ is a stationary probability
space. 
\end{definition}

In order to avoid degenerate games (for example, games where the least
expected loss is infinity, precluding any incentive to play the game),
Kalnishkan et al.\cite{KVV04} restricts the game in the following
manner.

\begin{itemize}
\item We restrict $\Gamma$ to be a compact space. For the binary
alphabet space, the prediction space is $[0,1]$. 
\item The loss function $\lambda$ is an extended real-valued convex
  function on $\Sigma \times \Gamma$. We take the discrete topology on
  the alphabet and the standard topology on $[0,1]$. Then $\lambda$ is
  continuous with respect to their product topology.
\item There is a prediction $\gamma \in \Gamma$ such that for every $b
  \in \Sigma$, the inequality $\lambda(b,\gamma) < \infty$ holds. This
  property ensures that the $n$-ary entropy is a finite quantity.
\item If there are $\gamma_0 \in \Gamma$ such that for some $b \in
  \Sigma$, the loss $\lambda(b, \gamma) = \infty$, then there is a
  sequence $\gamma_1, \gamma_2, \dots \to \gamma$ such that for
  each $\gamma_i$, we have $\lambda(b, \gamma_i) < \infty$.
\end{itemize}

A game which obeys these conditions is said to be \emph{regular}. The
last condition is necessary (but not sufficient) to ensure that
predictive complexity exists for the game. We need this property
crucially in Theorems \ref{smb} and \ref{conv_pred_comp}.

The $n$ step generalized entropy is the least expected loss incurred
by any strategy, on $\Sigma^n$.  Since $\Sigma^n$ is a compact space
and $\lambda$ is continuous in both its arguments, the infimum in the
above expression is attained by some strategy. \footnote{The authors
  remark in \cite{KVV04} that such a strategy need not exist for
  $\Sigma^*$.}

\begin{example}
The Log-Loss game: Consider the binary alphabet and predictions be
values in [0,1]. Let $p_0$ and $p_1$ be the probability of the bit 0
and bit 1, respectively.

Suppose we define the loss function by $\lambda(b, \gamma) = - \log
(\mid\overline{b}-\gamma\mid)$, where $b$ is a bit, $\overline{b}$ its
complement, and $\gamma \in [0,1]$. Then the minimal expected loss
over one bit is obtained at $\gamma = p_1$, ensuring that $H(p_1)$ is
the Shannon entropy of the distribution. \hfill (End of Example)
\end{example}

\begin{definition}
  \label{condentropy}
  The \emph{generalized conditional entropy} of $\Sigma^n$ given
  $\Sigma^m$ is defined as
\begin{align*}
H_{n|m}&=\inf_{\wp}\sum_{w\in\Sigma^m}P(w)\;\sum_{x\in\Sigma^n}P\{x\mid
w\}\sum_{i=0}^{m-1}\lambda\left(x_i,\wp^{i+m}(w \cdot
x^{i-1}_0)\right)\\ &=\inf_{\wp}\sum_{wx\in\Sigma^{n+m}}P(wx)
\sum_{i=0}^{m-1}\lambda\left(x_i,\wp^{i+m}(w \cdot x^{i-1}_0)\right)
\end{align*}
\end{definition}

This is an analogue of the definition of conditional Shannon
entropy. The inner term in Definition \ref{condentropy} can also be
expressed as follows.
$$\sum_{i=0}^{m-1}\lambda \left(x_i,\wp^{i+m}(w \cdot
x^{i-1}_0)\right) = \text{Loss}(wx,\wp)-\text{Loss}(w,\wp).$$

When we generalize the theory to handle arbitrary loss functions, we
do lose some ideal properties that Shannon entropy has. The following
theorem states that Shannon entropy is the unique function having
certain ideal properties that we desire in a measure of
information \cite{Khinchin57}. 

\begin{theorem}
\label{kut}
Suppose $F$ is a continuous function mapping $n$-dimensional
probability distributions to $[0,1]$ having the following properties.
\begin{enumerate}
\item For any random variables $A$ and $B$, $F(AB)=F(A) + F(B|A)$.
\item The $n$-dimensional uniform distribution has the largest entropy
  among $n$-dimensional distributions.
\item $F(p_1, p_2, \dots, p_n, 0) = F(p_1, p_2, \dots, p_n).$
\end{enumerate}
Then there is a positive constant $\lambda$ such that for every
$n$-dimensional probability vector $(p_1, \dots, p_n)$,
$H(p_1, p_2, \dots, p_n) = \lambda F(p_1, p_2, \dots, p_n)$.
\end{theorem}

With our definition of the cumulative loss, we can establish the chain
rule for generalized entropy.

\begin{lemma}
  \label{chainrule}
For all positive natural numbers $m$ and $n$, we have
$H_{m+n}=H_m+H_{n|m}$.
\end{lemma}
\begin{proof}
  In Definition \ref{condentropy}, $\wp^i$ for $0\le i\le m$ does not
  play any role in the infimum and likewise in Definition
  \ref{nentropy}, $\wp^i$ for $i\ge n$ does not play any role in the
  infimum $\inf$. This observation allows us to deduce that
\begin{multline}
    H_m+H_{n/m} =
    \inf_{\wp}\left(
    \sum_{w\in\Sigma^m}P(w)\sum_{x\in\Sigma^n}
    P\{x\mid w\}\sum_{i=0}^{m-1}\lambda\left(x_i,
    \wp^{i+m}(w \cdot x^{i-1}_0)\right)\right)+\\
    \;\inf_{\wp}\sum_{w\in\Sigma^m}P(w)\text{Loss}(w,\wp)\\
    =\inf_{\wp} \sum_{w\in\Sigma^m}P(w)\left(\sum_{x\in\Sigma^n}
    \sum_{i=0}^{m-1}\lambda\left(x_i, \wp^{i+m}(w \cdot
    x^{i-1}_0)\right)+
    \sum_{w\in\Sigma^m}\text{Loss}(w,\wp)\right).
\end{multline}
Now, 
\begin{align*}
    &\inf_{\wp}\sum_{w\in\Sigma^m}P(w)\left(\text{Loss}(w,\wp)+
    \sum_{w'\in\Sigma^n}P\{w'\mid w\}\sum_{i=0}^{m-1}
    \lambda(x_i,\wp^{i+m}(w \cdot x^{i-1}_0))\right)\\
    = & \inf_{\wp}\sum_{w\in\Sigma^m}P(w)
    \sum_{w'\in\Sigma^n}P\{w'\mid w\}
    \left(\text{Loss}(w,\wp)+\sum_{i=0}^{m-1}\lambda(x_i,\wp^{i+m}(w
    \cdot x^{i-1}_0))\right)\\ 
    = & \inf_{\wp}\sum_{w\in\Sigma^{m+n}}P(w)\text{Loss}(w,\wp) =  H_{m+n}.
    \end{align*}
\end{proof}

Since $\lambda$ is non-negative, it is clear that all entropies
defined so far are non-negative. An immediate consequence of this is
$H_{m+n}\ge H_m$ for all $m,n\ge 0$. We see that this style of proof
referring to strategies in games yields new intuitive proofs of such
inequalities. 

Since conditions 1 and 3 in Theorem \ref{kut} are satisfied,
Khinchin's uniqueness theorem therefore leads us to conclude that with
a generalized entropy, the uniform distribution need not have maximal
entropy - for example, the square-loss is not maximized at the uniform
distribution.

\section{Entropy of a Regular Game}
The goal of this section is to define the notion of the entropy of a
regular game. Our idea is to define it to be the limiting rate of the
$n$-step generalized entropies of the game. We now show that if the
game is regular and the probability distribution is stationary, such a
limit exists. Thus the notion of the entropy of a regular game is
well-defined.  

\begin{lemma}
  \label{shannonineq}
[Generalized Shannon Inequality] For any regular game and non-negative
integers $m$ and $n$, we have $H_{m/n}\le H_m$.
\end{lemma}
\begin{proof}
  The following proof is for $m=1$. In this special case 
$H_1=\displaystyle\inf_{\gamma\in\Gamma}\sum_{a\in\Sigma}P(a)\lambda(a,\gamma)$
and
\begin{align*}
H_{1|n} = \displaystyle\inf_{f\in
  \Pi^n}\sum_{w\in\Sigma^n}P(w)\sum_{a\in\Sigma}P\{a\mid
w\}\lambda(a,f(w)) = \displaystyle\inf_{f\in
  \Pi^n}\sum_{a\in\Sigma}P(a)\sum_{w\in\Sigma^n}P\{w\mid
a\}\lambda(a,f(w))
\end{align*}
Now pick the $\gamma\in \Gamma$ which matches $H_{1}$. We can do this
because regularity condition of game requires $\Gamma$ to be compact.
The loss function is continuous in both its arguments ensuring that
the expected loss in (\ref{eqn:nentropy}) is a continuous function on
a compact space. Now define $f':\Sigma^n\to \{\gamma\}$. Clearly,
$f'\in \Pi^n$.  So,
\begin{align*}
H_{1/n}  \le &
  \sum_{a\in\Sigma}P(a)\sum_{w\in\Sigma^n}P\{w\mid a\}\lambda(a,f'(w)) 
   = &
  \sum_{a\in\Sigma}P(a)\sum_{w\in\Sigma^n}P\{w\mid a\}\lambda(a,\gamma)\\ 
  = & \sum_{a\in\Sigma}P(a)\lambda(a,\gamma)  =  H_1
\end{align*}

The general case proceeds by induction by defining
$f'^{i+n}(w\ w'^{i-1}_0) = f^i(w'^{i-1}_0)$, where $w$ is an $n$-long
string and $1 \le i \le m$.
\end{proof}

In the special case of the log-loss game with a Bernoulli distribution
on the finite alphabet, the argument above yields a new argument for
the Shannon inequality.

\begin{lemma}
\label{shannonineq2}
For any regular game, any stationary distribution $P$ defined on it,
and any positive pair of natural numbers $m$ and $n$, $H_{m|n} \ge
H_{m|n+1}$.
\end{lemma}
\begin{proof}
We prove the inequality for $m=1$. The general case would follow from
application of Lemma \ref{chainrule}.  We have,
\begin{align*}
H_{1|n} = \inf_{f\in
  \Pi^n}\sum_{w\in\Sigma^n}P(w)\sum_{a\in\Sigma}P\{a\mid
w\}\lambda(a,f(w)) = \displaystyle\inf_{f\in
  \Pi^n}\sum_{a\in\Sigma}\sum_{w\in\Sigma^n}P\{wa\}\lambda(a,f(w))
\end{align*}
and similarly $H_{1/n+1}=\displaystyle\inf_{f'\in
  F^{n+1}}\sum_{a\in\Sigma}\sum_{w\in\Sigma^{n+1}}P\{wa\}\lambda(a,f'(w))$.

We show for each $f\in \Pi^n$ we have a $f'\in F^{n+1}$ which matches
the inner quantity on which infimum is taken. Then, by taking infimum
over $F^{n+1}$,we would have $H_{1/k}\ge H_{1/k+1}$. Fix a $f\in \Pi^n$
and consider $f'\in F^{n+1}$ defined as $f'(bw)=f(w)$ for all
$w\in\Sigma^n,b\in\Sigma$. Now,
\begin{eqnarray*}
\displaystyle\sum_{a\in\Sigma}\sum_{w\in\Sigma^{n+1}}P\{wa\}\lambda(a,f'(w))
&=&
\sum_{a\in\Sigma}\sum_{b\in\Sigma}\sum_{w'\in\Sigma^{n}}P\{bw'a\}\lambda(a,f'(bw'))
\\ &=&
\sum_{a\in\Sigma}\sum_{w'\in\Sigma^{n}}\sum_{b\in\Sigma}P\{bw'a\}\lambda(a,f(w'))
\\ &=& \sum_{a\in\Sigma}\sum_{w'\in\Sigma^{n}}P\{w'a\}\lambda(a,f(w'))
\\
\end{eqnarray*}
where the last step follows from stationarity of $P$ (i.e,
$\sum_{b\in\Sigma}P\{bw\}=P\{w\}$ for all $w\in\Sigma^n$).
\end{proof}

\begin{theorem}
\label{thm:entropy_exists}
 For any regular game $\mathcal{G}$ and stationary $(\Sigma^\infty,
 \mathcal{F}, P)$,
 $\displaystyle\lim_{n\to\infty}\frac{H_n}{n}$ exists and is finite.
\end{theorem}
\begin{proof}
  From the regularity condition, we get $H_1$ is finite. From Lemma
  \ref{chainrule}, it follows that $H_n=\sum_{i=0}^{n-1}H_{1|i}$.

  By Lemma \ref{shannonineq2}, $H_{1|k}\ge H_{1|(k+1)}$. Since
  entropies are non-negative, the sequence $\{H_{1|n}\}$ is a bounded,
  monotone decreasing sequence of reals. Hence, it has a limit which
  we denote by $H_{1|\infty}$. It also follows that $H_{1|\infty}$ is
  at most $H_1$.

So by Ces\`aro mean,
$\displaystyle\lim_{n\to\infty}\frac{H_n}{n}=$
$\lim_{n\to\infty}\frac{1}{n}\sum_{i=0}^{n-1}H_{1|i}=$
$\lim_{n\to\infty}H_{1|n}=H_{1|\infty}$. 
\end{proof}

\begin{definition}
  \label{entropy}
  Let $\mathcal{G}=(\Sigma^\infty,\Gamma,\lambda)$ be a regular game
  and $(\Sigma^\infty, \mathcal{F}, P)$ be a stationary
  distribution. Then The \emph{generalized entropy} of the game is
  defined as
  $$H=\displaystyle\lim_{n\to\infty}\frac{H_n}{n}.$$
\end{definition}

\section{A Shannon-McMillan-Breiman Theorem}

We now show that for regular games with a suitable restriction on the
loss functions, optimal processes exist and they attain the
generalized entropy rate of the stationary ergodic process. Our
approach to this result is through uniform integrability and the
Vitali Convergence theorem, which contrasts with the usual approach
using the Dominated Convergence Theorem. First, we define the notion
of a \emph{strongly regular game}, for which the result
holds. \footnote{Kalnishkan et al. \cite{KVV07} consider the notion of
  mixable games, which characterize regular games with optimality. In
  comparison, our conditions are based on integrability of the loss
  function.} We will derive two consequences of strong regularity,
\emph{viz.}
\begin{enumerate}
\item The existence of a limiting function for the loss function,
   $P$-almost everywhere.
\item The integrability of this limiting function
\end{enumerate}
We urilize these in the proof of the Shannon-McMillan-Breiman
Theorem. We conclude with two examples, illustrating that Theorem
\ref{smb} properly generalizes the classical Shannon-McMillan-Breiman
theorem.

\begin{definition}
Let $(\Omega, \mathcal{F}, P)$ be a probability space. A sequence of
functions $\{f_n\}^{\infty}_{n=1}$ is called \emph{uniformly
  integrable} if 
\begin{align}
\label{ui}
\lim_{\alpha \to \infty} \sup_{n} \int |f_n| I_{[|f_n| > \alpha]} dP =
0,
\end{align}
where $I_{[|f_n| > \alpha]}$ is the indicator function which is 1 at
points $\omega$ with $|f_n(\omega)| > \alpha$ and is 0 otherwise.
\end{definition}

If the sequence $\{f_n\}^{\infty}_{n=1}$ is uniformly integrable, then
for every $\epsilon>0$, and any large enough $\alpha$,
\begin{align}
\label{ui_unlimit}
\sup_n \int |f_n| dP \le \alpha + \epsilon
\end{align}
In addition to uniform integrability, we also need a continuity
requirement over the space of strategies. We now introduce this. The
next lemma characterizes $H_{1|n}$ in terms of the loss incurred by an
optimal strategy on $\Sigma^n$.

\begin{lemma}
  \label{localglobalswitch}
  \begin{eqnarray*}
    H_{1|n}  =  \inf_{f\in
      \Pi^n}\sum_{w\in\Sigma^n}P(w)\sum_{a\in\Sigma}P\{a\mid
    w\}\lambda(a,f(w)) = \sum_{w\in\Sigma^n}P(w)\inf_{f\in
      \Pi^n}\sum_{a\in\Sigma}P\{a\mid w\}\lambda(a,f(w))\\
  \end{eqnarray*}
\end{lemma}
\begin{proof}
Let $n$ be an arbitrary number. For any string $w$ of length $n$, $P(w)
\ge 0$, thus it follows that
\begin{align*}
\inf_{f\in \Pi^n}
\sum_{w\in\Sigma^n}P(w)\sum_{a\in\Sigma}P\{a\mid w\}\lambda(a,f(w))
    & \ge & 
\sum_{w\in\Sigma^n}P(w)\inf_{f\in \Pi^n}\sum_{a\in\Sigma}P\{a\mid
w\}\lambda(a,f(w)), 
\end{align*}
hence it suffices to prove that that the opposite inequality holds.

For each $n$-long string $w$, let $f_{w}$ be the function which
attains the infimum
$$\inf_{f\in \Pi^n}\sum_{a\in\Sigma}P\{a\mid w\}\lambda(a,f(w)).$$

Thus, the required expectation of infima can be written in terms of
these functions as
\begin{align*}
\sum_{w\in\Sigma^n}P(w)\inf_{f\in \Pi^n}\sum_{a\in\Sigma}P\{a\mid
w\}\lambda(a,f(w)) \;=\; \sum_{w \in \Sigma^n} P(w) \sum_{a \in
  \Sigma} P\{a\mid w\} \lambda(a,f_{w}(w)).
\end{align*}

We can now define a function $f: \Sigma^n \to \Sigma$ as 
$$f(w) = f_{w}(w), \qquad w\in\Sigma^n.$$ 
It is clear from the definition of the function that
\begin{align*}
\sum_{w\in\Sigma^n}P(w)\sum_{a\in\Sigma}P\{a\mid w\}\lambda(a,f(w))
      = 
\sum_{w\in\Sigma^n}P(w)\sum_{a\in\Sigma}P\{a\mid w\}\lambda(a,f_{w}(w)),
\end{align*}
which implies the desired inequality.
\end{proof}

Lemma \ref{localglobalswitch} lets us analyse loss incurred by some
``optimal'' strategy. From Lemma \ref{localglobalswitch}, we can see
given $w\in\Sigma^n$, optimal loss depends on the conditional
probability distribution $(P\{0\mid w\},P\{1\mid w\})$. Let
$s(P\{0\mid w\})$ be the strategy that gives optimal loss in
$H_{1|n}$.

Let us define the following functions on $\Sigma^\infty$.
\begin{align*}
g_k(\omega)&=\lambda(\omega_0, s(P\{0 \mid \omega^{-1}_{-k}\}))\\
g(\omega)&=\lambda(\omega_0,s(P\{0\mid w^{-1}_{-\infty}\})).
\end{align*}
So, 
$\text{Loss}(\omega^{n-1}_0,\wp_n) =
\displaystyle\sum_{k=0}^{n-1}g_k(T^k \omega)$. 


\begin{definition}
A regular game is \emph{strongly regular} if
\begin{enumerate}
\item $s$ is a continuous function of the conditional probability.
\item For each natural number $N$, define $G_N : \Omega \to
  [0,\infty]$ by 
$$G_N(\omega) = \sup_{k \ge N} \left|g_k(\omega) - g(\omega)\right|.$$
We require that $\{G_N\}^{\infty}_{N=1}$ is a uniformly integrable
sequence.
\end{enumerate}
\end{definition}

First, we explain a consequence of condition (1). For a stationary
ergodic distribution $P$, $P\{0 \mid \omega^{-1}_{-k}\} \to P\{0 \mid
\omega^{-1}_{-\infty}\} $ as $k \to \infty$, and since $g_k$ is a
continuous function of the conditional distribution by condition (1),
we have that $g_k \to g$ as $k \to \infty$, $P$-almost everywhere.

We now elicit some consequences of our assumption of uniform
integrability. For uniformly integrable sequences of functions, their
limit function is integrable even in the absence of any dominating
function. This is known as the \emph{Vitali Convergence Theorem}
\cite{FollandRA}.

\begin{theorem}
\label{vc}
Let $(\Omega, \mathcal{F}, P)$ be a probability space. If
$\{f_n\}_{n=1}^{\infty}$ is a sequence of uniformly integrable
functions such that $f_n \to f$ $P$-almost everywhere, then $f$ is
integrable and 
$$\lim_{n \to \infty} \int |f_n - f| dP = 0.$$
\end{theorem}

Vitali Convergence of $\{G_N\}_{N=1}^{\infty}$ will be required in the
final part of the proof of Theorem \ref{smb}. We first show that
uniform integrability of $\{G_N\}_{N=1}^{\infty}$ yields the
integrability of the optimal loss.

\begin{lemma}
\label{limitintegrable}
For a strongly regular game and a stationary distribution $P$, 
$$\lim_{n \to \infty} \int g_n\; dP = \int \lim_{n \to \infty} g_n\;
dP = \int g\; dP.$$
\end{lemma}

\begin{proof}
We know that for each $n \in \N$,
$$\int|g_n|\;dP = \int g_n dP = H_{1|n},$$
which exists for regular games and stationary distributions. Now, for
every n,
$$\int |g_n|\;dP = \int|g-g_n-g|\; dP  \ge \int |g| dP - \int |g-g_n|
dP.$$
Hence we have
\begin{align}
\label{i_three}
H = \lim_{n \to \infty} \int |g_n|\;dP \ge \int |g|dP - \liminf_{n \to
  \infty} \int |g - g_n| dP.
\end{align}
By the uniform integrability of $\{G_N\}_{N=1}^{\infty}$, we have that
$$\lim_{n \to \infty} \int |g-g_n| dP = 0.$$
Thus, by (\ref{i_three}), we have $H \ge \int |g| dP$.   
\end{proof}

Using uniform integrability and the notion of continuity, we can
introduce the setting for our Shannon-McMillan-Breiman Theorem.

For the sake of convenience, in the following proof, we will consider
two-way infinite sequences. However, the same theorem holds for
one-way sequences as well (see Chapter 13 of \cite{Bill65}). We
briefly mention the formal correspondence.

Let $(X,\mathcal{B},\mu)$ be a measure space with $T$ being a measure
preserving transform, not necessarily invertible. We construct a
measure preserving system
$(\hat{X},\hat{\mathcal{B}},$ $\hat{\mu},\hat{T})$ as follows.
\begin{itemize}
\item Define $\hat{X} = \{(x_i)_{i \in \N} \mid x_i\in T^{-i} X,
  Tx_{i+1}=x_i \text{ for all } i \in \N \}$
\item Let $\pi_j : \hat{X}\to T^{-i} X$ be the projection function
  which projects $j^{th}$ co-ordinate of an element of $\hat{X}$, i.e,
  $\pi_j(x)=x_j$.  Construct a $\sigma$ algebra $\mathcal{B}'$
  generated by sets of the form $\pi_i^{-1} T^{-i} E$, for all $i\in
  \N$, and $E\in\mathcal{B}$.
\item Let $\hat{\mu}(\pi_i^{-1} T^{-i} E)=\mu(E)$ for all
  $E\in\mathcal{B}$.
\item Complete $\mathcal{B}'$ with respect to $\hat{\mu}$ to get
  $\hat{\mathcal{B}}$.
\item Define $\hat{T}:\hat{X}\to\hat{X}$ by $\hat{T}((x_i)_{i \in
  \N})=((Tx_i)_{i \in \N})$.
\end{itemize} 
Clearly, $\hat{T}$ is an invertible transform given by
$\hat{T}^{-1}(x_1,x_2,x_3,\cdots)=(x_2,x_3,x_4,\cdots)$. Since $T$ is
measure preserving, $\hat{T}$ is also measure
preserving. $(\hat{X},\hat{\mathcal{B}},\hat{\mu},\hat{T})$ is called
\emph{natural extension of $(X,\mathcal{B},\mu,T )$}.  It is ergodic
iff the original system is ergodic. For unilateral alphebet system,
its natural extension has same entropy. For details, see Fact 4.3.2 of
\cite{entropyindynamicalsystem}.

\begin{theorem}
\label{smb}
For a strongly regular game $(\Sigma, \Gamma, \lambda)$, and
stationary ergodic distribution $(\Sigma^\infty, \mathcal{F}, P)$, let
H be the generalized entropy of the game. Moreover, let $\wp$ be a
strategy such that for every $n$, $\wp^n$ achieves $H_n$. Then for
$\omega \in\Omega$, the following holds:
\begin{align}
\label{eqn:smb}
\lim_{n\to\infty}\frac{\text{Loss}(\omega^{n-1}_0,\wp^n)}{n}=H
\end{align}
for $P$-almost every $\omega$.
\end{theorem}

We cannot use the Birkhoff's ergodic theorem (see for example,
\cite{Bill65}) directly to prove the above theorem, since the summands
in the Birkhoff average on the left of (\ref{eqn:smb}) depend in
general on $n$, and are not the same integrable function. We however
can use the convergence in conditional distributions ensured by a
stationary distribution, in conjunction with Birkhoff's ergodic
theorem to establish our result. 

\begin{proof}
Recall that $g_k \to g$ almost everywhere, and $\int g$ exists by
Lemma \ref{limitintegrable}. We know $\text{Loss}(\omega^{n-1}_0,
\wp^n) = g_n(\omega)$.

Since $T$ is measure preserving transformation, by change of variable,
$$\int_{\Omega}g_k(\omega)dP=\int_{\Omega}g_k(T^k\omega)dP=H_{1|k}.$$
Thus
$$\int g(w) dP = \lim_{n\to\infty} \int g_n(w) dP = \lim_{n\to\infty}
H_{1|n}=H.$$

By the Ergodic theorem, we get 
$$\lim_{n\to\infty}\frac{1}{n}\sum_{k=0}^{n-1}g(T^kw)=\int g(w) dP =
  H,$$
for $P$-almost every $\omega \in \Omega$.`

Now,
$$\frac{1}{n}\sum_{k=0}^{n-1}g_k(T^kw) =
\frac{1}{n}\sum_{k=0}^{n-1}g(T^kw) +
\frac{1}{n}\sum_{k=0}^{n-1}(g_k(T^kw)-g(T^kw)).$$ where the first term
tends to $H$ as $n \to \infty$. If we show second term in the previous
equation is tends to 0 a.e. as $n \to \infty$, we are done.

Define $G_N(w)=\sup_{k\ge N} | g_k(w)-g(w)|$. By the assumption of
strong regularity, the sequence of functions $\{G_N\}^{\infty}_{N=1}$
is uniformly integrable. Also, since $g_n \to g$ $P$-a.e., we know
that $G_N \to 0$ $P$-almost everywhere as $N \to \infty$. By the
Vitali Convergence Theorem,
$$\lim_{N \to \infty} \int G_N\; dP = \int \lim_{N \to \infty} G_N\; dP =
0.$$

Now for each $N$,
\begin{align*}
  \displaystyle \limsup_{n\to\infty}
  \left|\frac{1}{n}\sum_{k=0}^{n-1}(g_k(T^k\omega)-g(T^k\omega)\right|
  &\le
  \limsup_{n\to\infty}\frac{1}{n}\sum_{k=0}^{n-1}|g_k(T^k\omega)-g(T^k\omega)|\\ &\le
  \limsup_{n\to\infty}\frac{1}{n}\sum_{k=0}^{n-1}G_N(T^k\omega) = \int
  G_N(\omega) dP
\end{align*}
where the last equality follows from Birkhoff Ergodic Theorem. Note
that this holds for all values of $N$ and right side converges to $0$
a.e. as $N\to\infty$. Since the left side is non-negative, it is $0$
a.e. So, $\frac{1}{n}\sum_{k=0}^{n-1}(g_k(T^k\omega)-g(T^k\omega))\to
0$ as $n\to\infty$. This concludes the proof.
\end{proof}

Recall that the generalized entropy of the log-loss game is the
Shannon entropy.  We now show the square loss and the log-loss games
are strongly regular, thus establishing that we have a proper
generalization of the classical Shannon-McMillan-Breiman theorem.

\begin{example}
Log-loss Game. The loss function $\lambda: \{0,1\} \times [0,1] \to
[0,\infty]$ is defined by
$$\lambda(b, \gamma) = - \log (|b - \gamma|).$$

The optimal strategy is given by $s_k = P\{0 \mid \omega^{-1}_{-k}\}$,
which is a continuous function of the conditional probability.

We have that for any $N$,
\begin{align*}
\int \sup_{k \ge N} |g_k(\omega) - g(\omega)| dP &\le 
\int \sup_{n \ge 1} |g_n(\omega) - g(\omega)| dP
&\le \int \sup_{n \ge 1} |g_n(\omega)| + \int g dP.
\end{align*}

Hence to show that the sequence $\sup_{k \ge N} |g_k(\omega) -
g(\omega)|$ is uniformly integrable, it suffices to show that 
$$\int \sup_{n \ge 1} |g_n(\omega)| dP$$ is integrable. It is easy to
show that for a stationary distribution $P$ and any $r \in \R$,
$$P\{\omega \mid \sup_k |g_k(\omega)| \ge r\} \le 2 e^{-r},$$
from which the integrability of $\sup_k g_k$ follows. 

Thus $\sup_{k \ge N} |g_k - g|$, for $N=1,2, \dots$ forms a uniformly
integrable sequence of functions, and Theorem \ref{smb} holds for the
log-loss game.
\end{example}

\begin{example}
Square-loss game. The loss function in the square loss game 
$\lambda: \{0,1\} \times [0, 1] \to [0,1]$ defined by
\begin{align}
\lambda(b,\gamma) = (b - \gamma)^2.
\end{align}
The optimal strategy in the square-loss game is to pick $\gamma =
p\{1 \mid \omega^{-1}_{-k}\} $, which is continuous in the conditional
probability. 

This loss function is bounded, hence 
$$\int \sup_{k \ge 1} \left| g_n(\omega) - g(\omega) \right| dP \le 
  \int 1 dP = 1,$$
ensuring that $G_N = \sup_{k \ge N} |g_k(\omega) - g(\omega)|$ is
uniformly integrable. Thus Theorem \ref{smb} holds for the square-loss
game. 
\end{example}

\section{Predictive Complexity of Stationary Ergodic Games}
\newcommand{\sscr}{\mathcal{S}}

We now consider computable prediction strategies. We would like to
define the inherent unpredictability of a string $x$ as the
performance of an optimal computable predictor on $x$. It is not clear
that one such predictor exists for any game. The work of Vovk and
Watkins\cite{VW98} establishes a sufficient condition for predictive
complexity to exist.

\begin{definition}
A pair of points $(s_0, s_1) \in (-\infty, \infty]^2$ is called a
  \emph{superscore}\footnote{In \cite{KVV04}, \cite{KVV07}, the
    concept is called a superprediction.} if there is a prediction
  $\gamma \in \Gamma$ such that $ \lambda(0, \gamma) \le s_0 \text{
    and } \lambda(1, \gamma) \le s_1$.  We denote the set of
  superscores for a regular game $\mathcal{G}$ by $\sscr$.
\end{definition}

\begin{definition}
A prediction strategy $\wp: \Sigma^* \to (-\infty, \infty]$ is called
  a \emph{superloss process} if the following conditions hold.
\begin{enumerate}
\item $\wp(\Lambda) = 0$
\item For every string $x$, the pair $(\wp(x0) - \wp(x), \wp(x1) -
  \wp(x))$ is a superscore with respect to the game.
\item $\wp$ is upper semicomputable.
\end{enumerate}
\end{definition}

A superloss process $K$ is \emph{universal} if for any superloss
process $\wp$ there is a constant $C$ such that for every string $x$,
$$ K(x) \le \wp(x) + C.$$
It follows that the difference in loss between any two superloss
processes is bounded by a constant. Hence we may pick a particular
superloss process $\mathcal{K}$ and call $\mathcal{K}(x)$ the
\emph{predictive complexity} of the string $x$ with respect to the
game $\mathcal{G}$.

When we consider regular games, it is not necessary that an optimal
strategy exists on $\Sigma^*$ which incurs at most an additive loss
when compared to any other prediction process. However, Vovk
\cite{Vovk98} and Vovk and Watkins\cite{VW98} introduced the concept
of \emph{mixability} to ensure that one such universal process exists.

\begin{definition}
Let $\beta \in (0,1)$. Consider the homeomorphism $h_\beta : (-\infty,
\infty]^2 \to [0, \infty)^2$ specified by
$h_\beta(x,y) = (\beta^x, \beta^y)$.
A regular game $\mathcal{G}$ with set of superscores $\sscr$ is called
\emph{$\beta$-mixable} if the set $h_\beta(\sscr)$ is convex. A game
$\mathcal{G}$ is called \emph{mixable} if it is $\beta$-mixable for
some $\beta \in (0,1)$.
\end{definition}

\begin{theorem} \cite{VW98}
If a game $\mathcal{G}$ with set of superscores $\sscr$ is
mixable, then $\mathcal{G}$ has a predictive complexity.
\end{theorem}

It is known that the logloss and the square loss games are
mixable. The coincidence of logloss and Kolmogorov complexity enables
us to view predictive complexity as a generalization of predictive
complexity. Absolute loss game is known not to be mixable
\cite{Vyugin02}.

We mention a loss bound which holds for mixable games. This is used in
the proof of the theorem which follows.

\begin{lemma}\cite{KVV04}
\label{loss_bound}
If $\K$ is predictive complexity of a mixable game $\mathcal{G}$, then
there is a positive constant c such that $|\K(xb)-\K(x)|\le c \ln n$
for all $n=1,2,\cdots$, strings $x$ and bits $b$.
\end{lemma}

We can now show that for a strongly regular mixable game
$\mathcal{G}$, the predictive complexity rate on an infinite sequence
of outcomes attains the generalized entropy of the stationary ergodic
distribution $P$, almost everywhere.

\begin{theorem}
\label{conv_pred_comp}
  Let $\mathcal{G}=(\Omega,\Gamma,\lambda)$ be a strongly regular
  mixable game with predictive complexity $\mathcal{K}$. Let $(\Omega,
  \mathcal{F}, P)$ be the probability space over the outcomes where
  $P$ is a stationary ergodic distribution with generalized entropy
  $H$.  Then
    $$\lim_{n\to\infty}\frac{\mathcal{K}(\omega^{n-1}_0)}{n}=H,$$
for $P$-almost every $\omega \in \Omega$.
\end{theorem}
\begin{proof}

(A) Upper Bound: First we show that
  $\lim_{n\to\infty}\frac{\mathcal{K}(w^{n-1}_0)}{n}<H+\epsilon$
  for any $\epsilon>0$. This is an application of our
  Shannon-McMillan-Breiman theorem, Theorem \ref{smb} for generalized
  entropy.

Let $\wp_n$ be the strategy which achieves $H_{1|n}$. There is a
computable strategy $\zeta$ so that for all $0\le i\le n-1$,
$$\lambda(a,\zeta_i(w)) < \lambda(a,\wp_n^i(w))+\frac{\epsilon}{2}$$
for all $a\in \Sigma$ and for all $w \in \Sigma^{i}$. This is possible
since set of all such strategies constitute an open set. By the
definition of predictive complexity, we have
    \begin{eqnarray*} 
      \mathcal{K}(\omega^{n-1}_0) & \le &
      \text{Loss}(w^{n-1}_0,\zeta)+O(1)\\ & \le &
      \text{Loss}(w^{n-1}_0,\wp_n) + \frac{\epsilon n}{2}
      +O(1)\\
    \end{eqnarray*}
 By the Shannon-McMillan-Breiman Theorem, for large enough $n$,
\begin{align*}
      \text{Loss}(w^{n-1}_0,\wp_n) + \frac{\epsilon n}{2} +O(1)
      \;\le\; H + \epsilon O(n) + \frac{\epsilon n}{2} + O(1).\\
    \end{align*}
Taking limits as $n \to \infty$, we have that
 $$\lim_{n\to\infty}\frac{\mathcal{K}(w^{n-1}_0)}{n}<H+\epsilon.$$

(B) We now establish the reverse inequality,
$\lim_{n\to\infty}\frac{\mathcal{K}(\omega^{n-1}_0)}{n}>H-\epsilon$
for $\epsilon>0$.  Since
    $$(\K(\omega^{n-1}_0\cdot0)-\K({\omega^{n-1}_0}),\;\;\K(\omega^{n-1}_0\cdot
   1)-\K({\omega^{n-1}_0}))$$ is a superscore, we have
   $E(\eta_n|\omega^{n-1}_0)\ge H_{1|n}$ where $\eta_n=\K(\omega^{n-1}_0)-\K(\omega^{n-1}_0)$.

Now we can apply the martingale strong law of large numbers, Theorem
VII.5.4 of \cite{Shiryaev} and get 

\begin{eqnarray*}
\frac{\K(\omega^{n-1}_0)}{n} =&
\frac{1}{n}\sum_{i=0}^{n-1}\eta_i &=
\frac{1}{n}\sum_{i=0}^{n-1}E(\eta_i|\omega^{i-1}_0) + o(1)\\ 
\ge& \frac{1}{n}\sum_{i=0}^{n-1}H_{1|n} + o(1) &= H +o(1),
\end{eqnarray*}

where the last equality is obtained by Theorem
\ref{thm:entropy_exists}. 
\end{proof}

\section*{Acknowledgments} 

The authors would like to thank John Hitchcock and Vladimir
  V'yugin for helpful discussions.

\bibliography{main,dim,random}

\begin{thebibliography}{10}

\bibitem{Bill65}
P.~Billingsley.
\newblock {\em Ergodic Theory and Information}.
\newblock John Wiley \& Sons, 1965.

\bibitem{CBL06}
N.~Cesa-Bianchi and G.~Lugosi.
\newblock {\em Prediction, Learning and Games}.
\newblock Cambridge University Press, 2006.

\bibitem{entropyindynamicalsystem}
T.~Downarowicz.
\newblock {\em Entropy in Dynamical Systems}.
\newblock New Mathematical Monographs. Cambridge University Press, 2011.

\bibitem{Fede91}
M.~Feder.
\newblock Gambling using a finite state machine.
\newblock {\em IEEE Transactions on Information Theory}, 37:1459--1461, 1991.

\bibitem{FeMeGu92}
M.~Feder, N.~Merhav, and M.~Gutman.
\newblock Universal prediction of individual sequences.
\newblock {\em IEEE Transations on Information Theory}, 38:1258--1270, 1992.

\bibitem{FollandRA}
Gerald~B. Folland.
\newblock {\em Real Analysis}.
\newblock Wiley, 1999.

\bibitem{Fortnow:PD}
L.~Fortnow and J.~H. Lutz.
\newblock Prediction and dimension.
\newblock {\em Journal of Computer and System Sciences}, 70:570--589, 2005.

\bibitem{GD04}
P.~D. Gr{\"u}nwald and A.~P. Dawid.
\newblock Game theory, maximum entropy, minimum discrepancy and robust bayesian
  decision theory.
\newblock {\em Annals of Statistics}, 32(4):1367--1433, 2004.

\bibitem{KVV07}
Y.~Kalnishkan, V.~Vovk, and M.~V. Vyugin.
\newblock Generalized entropies and asymptotic complexities of languages.
\newblock In {\em Learning Theory, 20th Annual Conference on Learning Theory},
  pages 293--307, 2007.

\bibitem{KVV04}
Yuri Kalnishkan, Volodya Vovk, and Michael~V. Vyugin.
\newblock Loss functions, complexities, and the legendre transformation.
\newblock {\em Theor. Comput. Sci.}, 313(2):195--207, 2004.

\bibitem{Khinchin57}
A.~Ya. Khinchin.
\newblock {\em Mathematical Foundations of Information Theory}.
\newblock Dover Publications, 1957.

\bibitem{MeFe98}
N.~Merhav and M.~Feder.
\newblock Universal prediction.
\newblock {\em IEEE Transactions on Information Theory}, 44(6):2124--2147,
  1998.

\bibitem{Shiryaev}
A.~N. Shiryaev.
\newblock {\em Probability}.
\newblock Graduate Texts in Mathematics v.95. Springer, 2 edition, 1995.

\bibitem{Vovk98}
V.~Vovk.
\newblock A game of prediction with expert advice.
\newblock {\em Journal of Computer and System Sciences}, pages 153--173, 1998.

\bibitem{VW98}
V.~G. Vovk and Chris Watkins.
\newblock Universal portfolio selection.
\newblock In {\em COLT}, pages 12--23, 1998.

\bibitem{VV02}
Michael~V. Vyugin and Vladimir~V. V'yugin.
\newblock Predictive complexity and information.
\newblock In {\em COLT}, pages 90--104, 2002.

\bibitem{Vyugin02}
Vladimir V'yugin.
\newblock Suboptimal measures of predictive complexity for absolute loss
  function.
\newblock {\em Information and Computationi}, 175:146--157, 2006.

\end{thebibliography}

\end{document}